\documentclass[pdflatex]{sn-jnl}

\usepackage{graphicx}
\usepackage{multirow}
\usepackage{amsmath,amssymb,amsfonts}
\usepackage{amsthm}
\usepackage[title]{appendix}
\usepackage{xcolor}
\usepackage{textcomp}
\usepackage{manyfoot}
\usepackage{booktabs}
\usepackage[numbers,sort&compress]{natbib}
\newcommand{\doi}[1]{\url{https://doi.org/#1}}

\theoremstyle{thmstyleone}
\newtheorem{theorem}{Theorem}[section]
\newtheorem{proposition}[theorem]{Proposition}
\newtheorem{lemma}[theorem]{Lemma}
\newtheorem{corollary}[theorem]{Corollary}

\theoremstyle{thmstyletwo}

\newtheorem{remark}{Remark}

\theoremstyle{thmstylethree}
\newtheorem{definition}{Definition}[section]

\newcommand{\R}{\mathbb{R}}
\newcommand{\Z}{\mathbb{Z}}
\newcommand{\N}{\mathbb{N}}
\newcommand{\E}{\mathbb{E}}
\newcommand{\Prob}{\mathbb{P}}
\DeclareMathOperator{\Var}{Var}
\DeclareMathOperator{\Cov}{Cov}
\newcounter{maintablecount}

\begin{document}

\title[Weak Convergence for an Origin-Invariant CvM Statistic]{Weak Convergence and Gaussian Limits for a General-Dimensional Origin-Invariant Cram\'er--von Mises Statistic}

\author*[1]{\fnm{Marco} \sur{Mandap}\,}\email{marco.mandap@bulsu.edu.ph}

\affil*[1]{\orgdiv{Mathematics Department, College of Science}, \orgname{Bulacan State University}, \orgaddress{\city{City of Malolos}, \state{Bulacan}, \postcode{3000}, \country{Philippines}}}

\abstract{We introduce a general-$d$ origin-invariant Cram\'er--von Mises statistic $\bar\omega^2_d$ for testing complete spatial randomness (CSR), defined as the $2^d$-corner average of corner-oriented empirical-distribution-function discrepancies on $[0,1]^d$. The statistic admits a closed-form $O(n^2)$ computing formula. At $d=1$ it reduces to the classical rank-based CvM statistic, and at $d=2$ its computing formula coincides with Zimmerman's origin-invariant statistic. Under iid CSR, the associated empirical process converges in $\ell^\infty$ to a centered Gaussian process with an explicit cross-corner covariance kernel. The continuous mapping theorem yields a quadratic Gaussian limit governed by a positive trace-class covariance operator whose trace is $2^{-d}-3^{-d}$. A Poisson construction connects the conditional and unconditional formulations of CSR. For strictly stationary $\alpha$-mixing sequences with uniform marginals, we establish covariance summability, the long-run variance limit, and a finite-dimensional all-corners Gaussian limit. Monte Carlo experiments for $d=1,\dots,5$ illustrate null calibration, power against selected alternatives, and applications to three spatial datasets. The statistic is consistent against fixed iid alternatives whose distribution functions differ from uniformity on a set of positive measure; it is not claimed to be omnibus against every form of point-process interaction.}



\maketitle

\section{Introduction}\label{sec:intro}

The Cram\'er--von Mises (CvM) statistic \cite{Cramer1928,vonMises1931,Smirnov1937} is a classical goodness-of-fit criterion based on the integrated squared discrepancy between the empirical distribution function and the theoretical one. For testing that $n$ iid univariate samples $U_1,\dots,U_n$ are $\mathrm{Uniform}([0,1])$, the statistic
\[
  W^2 \;=\; n\int_0^1 \bigl(F_n(t) - t\bigr)^2\,dt
\]
where $F_n$ is the empirical CDF, admits the rank-based computing formula $W^2 = \frac{1}{12n} + \sum_{i=1}^n\bigl(\frac{2i-1}{2n} - U_{(i)}\bigr)^2$ and has the well-known asymptotic null distribution $\sum_{k=1}^\infty Z_k^2/(k^2\pi^2)$ with $\{Z_k\}$ iid $N(0,1)$ \cite{AndersonDarling1952,DurbinKnott1972}.

For spatial data---a mapped pattern of $n$ events in a rectangle $D\subset\R^2$---the natural multivariate generalization of the CvM statistic faces a subtlety absent from the univariate case: the empirical distribution function depends on which corner of $D$ is identified as the origin. Zimmerman \cite{Zimmerman1993} addressed this by defining an \emph{origin-invariant} bivariate CvM statistic $\bar\omega^2$, averaging the single-corner CvM functional over the four corners of $[0,1]^2$. Chiu and Liu \cite{ChiuLiu2009} subsequently proposed a generalized multivariate CvM statistic averaging over only the even-parity corners (the $2^{d-1}$ corners of $[0,1]^d$ with an even number of reflected coordinates), and studied its use in goodness-of-fit testing for multivariate distributions via Monte Carlo $p$-value approximation.

This paper constructs and analyzes a general-$d$ origin-invariant CvM test statistic $\bar\omega^2_d$, defined as the full $2^d$-corner orbit average. It establishes the statistic's computability, Gaussian-process limit, dependent-observation extension, and simulation-based null calibration. The contributions are:

\begin{enumerate}
\item \textbf{Definition and $O(n^2)$ computing formula.} We define $\bar\omega^2_d$ as the $2^d$-corner average of a corner-oriented CvM functional and prove a closed-form $O(n^2)$ computing formula valid for every corner and for the assembled orbit average (Theorems~\ref{thm:computing} and \ref{thm:omegaBar-computing}). The formula is directly implementable in statistical software for arbitrary $d$.
\item \textbf{Regression theorems.} We prove that $\bar\omega^2_d$ at $d=1$ reduces exactly to the classical rank-based CvM statistic (Theorem~\ref{thm:d1-regression}) and at $d=2$ is exactly Zimmerman's $\bar\omega^2$ (Theorem~\ref{thm:d2-regression}), despite the two constructions averaging over superficially different corner symmetry groups.
\item \textbf{Finite-dimensional Gaussian limit.} Under the null, the normalized empirical process at any finite selection of corners and evaluation points converges jointly in distribution to a mean-zero Gaussian with an explicit cross-corner covariance kernel (Theorem~\ref{thm:fidi-clt}). The existence of a limiting Gaussian process on the full index set is established via a Kolmogorov extension argument (Theorem~\ref{thm:limit-process}).
\item \textbf{Process-level weak convergence.} The corner-oriented indicator class is a finite union of bounded VC classes. Standard empirical-process theory therefore yields asymptotic equicontinuity and weak convergence in $\ell^\infty$ to a tight Gaussian limit (Theorems~\ref{thm:process-outer-wc} and~\ref{thm:equicontinuity}).
\item \textbf{Strictly stationary mixing extension.} For uniformly distributed, strictly stationary $\alpha$-mixing observations, we prove absolute summability of the all-corners covariance array, identify the positive semidefinite long-run covariance matrix, and obtain a finite-dimensional all-corners Gaussian limit under the corresponding scalar mixing-CLT conditions (Theorems~\ref{thm:mixing-cov-summable}--\ref{thm:mixing-joint-clt}).
\item \textbf{CSR$\leftrightarrow$Poisson bridge.} An ordered-list representation of a homogeneous Poisson point process on $[0,1]^d$ is constructed, and the law obtained by restricting it to a positive-probability sub-region is identified as a lower-rate Poisson process with the appropriate conditional location law (Theorems~\ref{thm:poisson-process} and \ref{thm:thinning}).
\item \textbf{Numerical assessment.} Monte Carlo null percentiles, size estimates, and selected power comparisons are reported for $d\le5$. They provide finite-sample evidence for the proposed calibration, but the comparisons remain specific to the stated simulation designs.
\end{enumerate}

Finite-dimensional convergence, asymptotic equicontinuity, total boundedness of the index semimetric, and continuity of the Gaussian limit yield full weak convergence by the standard empirical-process transfer theorem. Applying the continuous mapping theorem to the integrated square then gives the Karhunen--Lo\`eve null law described in Section~\ref{sec:null-law}.

\section{Preliminaries and notation}\label{sec:prelim}

\subsection{Corners and the corner-oriented EDF}\label{sec:corners}

For $d\ge 1$, a \emph{corner} of $[0,1]^d$ is a function $c : \{0,\dots,d-1\} \to \{\mathrm{true},\mathrm{false}\}$ specifying, for each coordinate $k$, whether the corner is the ``high'' ($c(k)=\mathrm{true}$) or ``low'' ($c(k)=\mathrm{false}$) endpoint. There are $2^d$ corners. The base corner $c_0$ is the all-false corner (the standard ``lower-left'' origin).

For a corner $c$ and a point $x\in[0,1]^d$, the \emph{corner-$c$ lower set} is
\[
  \mathrm{cornerLe}(c, x) \;=\; \{s\in[0,1]^d : \forall k,\; \text{if } c(k) \text{ then } s_k \ge x_k \text{ else } s_k \le x_k\}.
\]
The \emph{corner-$c$ empirical distribution function} (EDF) of a sample $X_1,\dots,X_n\in[0,1]^d$ is
\[
  F_n^c(x) \;=\; \frac{1}{n}\sum_{i=1}^n \mathbf{1}\{X_i \in \mathrm{cornerLe}(c,x)\},
\]
and the theoretical CDF under $\mathrm{Uniform}([0,1]^d)$ is
\[
  F^c(x) \;=\; \prod_{k=0}^{d-1} \bigl(\text{if } c(k) \text{ then } 1-x_k \text{ else } x_k\bigr).
\]

\subsection{The corner-oriented CvM functional}\label{sec:cvm-corner}

The \emph{corner-$c$ CvM functional} is
\begin{equation}\label{eq:cvmAtCorner}
  \omega^2_c \;=\; \int_{[0,1]^d} \bigl(F_n^c(x) - F^c(x)\bigr)^2\,dx.
\end{equation}
The \emph{general-$d$ origin-invariant statistic} is the $2^d$-corner orbit average:
\begin{equation}\label{eq:omegaBarSqD}
  \bar\omega^2_d \;=\; \frac{1}{2^d}\sum_{c\in\mathrm{Corner}(d)} \omega^2_c.
\end{equation}
This is the direct generalization of Zimmerman's $d=2$ construction \cite{Zimmerman1993} to arbitrary $d$. The averaging over all $2^d$ corners (not just the $2^{d-1}$ even-parity corners as in Chiu--Liu \cite{ChiuLiu2009}) ensures origin-invariance: the statistic does not depend on which corner of $[0,1]^d$ is identified as the origin.

\subsection{The algebraic relationship to Chiu--Liu's symmetrization}\label{sec:chiu-liu}

Chiu and Liu \cite{ChiuLiu2009} define the symmetric discrepancy $D_2^{(S)}$ as the average over only the even-parity corners ($2^{d-1}$ of the $2^d$). The full-orbit sum over all $2^d$ corners decomposes as the disjoint union of the even-parity sum and the odd-parity sum, and both halves have cardinality $2^{d-1}$ for $d\ge 1$ (proved via a coordinate-flip involution, not induction on $d$). Consequently, the full-orbit average is exactly the arithmetic mean of Chiu--Liu's even-parity average and the symmetric odd-parity average:
\[
  \bar\omega^2_d \;=\; \frac{1}{2}\bigl(D_2^{(S)} + D_2^{(S,\mathrm{odd})}\bigr).
\]

\begin{proposition}\label{prop:card-even}
For $d\ge 1$, the even-parity corners have cardinality $2^{d-1}$.
\end{proposition}

\begin{proof}
Flipping any one coordinate toggles parity, giving an explicit bijection between even- and odd-parity corners. Since even and odd together partition all $2^d$ corners, $2\cdot|\text{even}| = 2^d$.
\end{proof}

\section{Main results}\label{sec:results}

\subsection{The \texorpdfstring{$O(n^2)$}{O(n-squared)} computing formula}\label{sec:computing}

The corner-oriented CvM functional \eqref{eq:cvmAtCorner} admits a closed-form finite-sample computing formula, obtained by a Fubini reduction to per-coordinate integrals and explicit antiderivative computations on $[0,1]$.

\begin{definition}\label{def:factors}
For $a,b\in[0,1]^d$ and corner $c$, define:
\begin{align}
  \mathrm{pairFactor}_c(a,b) &= \prod_{k=0}^{d-1}\bigl(1-\max(\mathrm{reflect}_c(a_k),\mathrm{reflect}_c(b_k))\bigr),\\
  \mathrm{crossFactor}_c(a) &= \prod_{k=0}^{d-1}\frac{1-\mathrm{reflect}_c(a_k)^2}{2},
\end{align}
where $\mathrm{reflect}_{c,k}(t) = \text{if } c(k) \text{ then } 1-t \text{ else } t$; the subscript $k$ is suppressed inside the products above.
\end{definition}

\begin{theorem}[Per-corner computing formula]\label{thm:computing}
For any corner $c$, any $n\ge 1$, and any sample $X_1,\dots,X_n\in[0,1]^d$,
\[
  \omega^2_c \;=\; \frac{1}{n^2}\sum_{i,j=1}^n \mathrm{pairFactor}_c(X_i,X_j) \;-\; \frac{2}{n}\sum_{i=1}^n \mathrm{crossFactor}_c(X_i) \;+\; \left(\frac{1}{3}\right)^d.
\]
\end{theorem}

\begin{proof}
The integral \eqref{eq:cvmAtCorner} expands as $\int(F_n^c)^2 - 2F_n^c F^c + (F^c)^2$. By Fubini (the integrand is a product of per-coordinate indicator functions over $[0,1]^d$), each term reduces to a product of one-dimensional integrals. The three underlying evaluations are:
\begin{align}
  \int_{[0,1]^d} (F^c(x))^2\,dx &= \left(\frac{1}{3}\right)^d, \label{eq:int-sq}\\
  \int_{[0,1]^d} \mathbf{1}\{a\in\mathrm{cornerLe}(c,x)\}\,F^c(x)\,dx &= \mathrm{crossFactor}_c(a), \label{eq:int-ind-F}\\
  \int_{[0,1]^d} \mathbf{1}\{a\in\mathrm{cornerLe}(c,x)\}\,\mathbf{1}\{b\in\mathrm{cornerLe}(c,x)\}\,dx &= \mathrm{pairFactor}_c(a,b). \label{eq:int-ind-ind}
\end{align}
For the base corner $c_0$, each per-coordinate integral is a direct antiderivative computation on $[0,1]$. For a reflected coordinate ($c(k)=\mathrm{true}$), the substitution $t\mapsto 1-t$ (since $[0,1]$ is symmetric under this map) reduces to the base-corner computation. Combining the three evaluations gives the stated formula.
\end{proof}

Assembling the per-corner formula over the $2^d$-corner orbit average gives the computing formula for $\bar\omega^2_d$ itself:

\begin{theorem}[Orbit-averaged computing formula]\label{thm:omegaBar-computing}
For any $n\ge 1$ and any sample $X_1,\dots,X_n\in[0,1]^d$,
\[
  \bar\omega^2_d \;=\; \frac{1}{n^2}\sum_{i,j=1}^n \mathrm{pairFactorAvg}(X_i,X_j) \;-\; \frac{2}{n}\sum_{i=1}^n \mathrm{crossFactorAvg}(X_i) \;+\; \left(\frac{1}{3}\right)^d,
\]
where $\mathrm{pairFactorAvg}(a,b) = \frac{1}{2^d}\sum_c \mathrm{pairFactor}_c(a,b)$ and $\mathrm{crossFactorAvg}(a) = \frac{1}{2^d}\sum_c \mathrm{crossFactor}_c(a)$.
\end{theorem}

\begin{proof}
Unfold the orbit average \eqref{eq:omegaBarSqD}, substitute Theorem~\ref{thm:computing} at each corner, commute the corner-sum inward past the sample-index sums, and divide by $2^d$.
\end{proof}

\subsection{Regression to the classical \texorpdfstring{$d=1$}{d=1} CvM}\label{sec:d1-regression}

At $d=1$, the general-$d$ framework reduces to the classical rank-based CvM statistic. The double-sum/cross-term expression (unordered, no reference to order statistics) is shown to equal the rank-based sum via a combinatorial double-counting identity.

\begin{lemma}\label{lem:double-count}
For $a_0\le a_1\le\cdots\le a_{n-1}$,
\[
  \sum_{k=0}^{n-1}\sum_{l=0}^{n-1}\max(a_k,a_l) \;=\; \sum_{m=0}^{n-1}(2m+1)\,a_m.
\]
\end{lemma}

\begin{proof}
For each $m$, $a_m$ is the maximum of the pair $(k,l)$ whenever at least one index equals $m$ and the other is $\le m$. There are $2m+1$ such pairs.
\end{proof}

\begin{theorem}[$d=1$ regression]\label{thm:d1-regression}
For $n$ samples $Y_1,\dots,Y_n\in[0,1]$,
\[
  n\cdot\omega^2_{c_0} \;=\; \frac{1}{12n} + \sum_{i=0}^{n-1}\left(\frac{2i+1}{2n} - Y_{(i)}\right)^2,
\]
where $Y_{(0)}\le\cdots\le Y_{(n-1)}$ are the order statistics.
\end{theorem}

\begin{proof}
Theorem~\ref{thm:computing} at $d=1$, $c=c_0$ gives
\[
  n\omega^2_{c_0}
  =\sum_{i=1}^nY_i^2-\frac1n\sum_{i,j=1}^n\max(Y_i,Y_j)+\frac n3.
\]
By Lemma~\ref{lem:double-count}, the double sum equals
$\sum_{m=0}^{n-1}(2m+1)Y_{(m)}$. Expanding
$\sum_i\bigl((2i+1)/(2n)-Y_{(i)}\bigr)^2$ and using
$\sum_{m=0}^{n-1}(2m+1)^2=(4n^3-n)/3$ gives the displayed identity, including the constant $1/(12n)$.
\end{proof}

\begin{corollary}[Orbit-average regression at $d=1$]\label{cor:d1-orbit}
For the same sample,
\[
  n\bar\omega^2_1
  = \frac{1}{12n} + \sum_{i=0}^{n-1}\left(\frac{2i+1}{2n}-Y_{(i)}\right)^2.
\]
\end{corollary}

\begin{proof}
The high-corner functional becomes the base-corner functional of the reflected sample $1-Y_1,\dots,1-Y_n$ under the substitution $x\mapsto1-x$. Its order statistics are $1-Y_{(n-1-i)}$, and substitution into the rank formula shows that reflection leaves the classical CvM statistic unchanged. Thus the two corner functionals are equal, so their average equals either one.
\end{proof}

\subsection{Regression to Zimmerman's \texorpdfstring{$d=2$}{d=2} statistic}\label{sec:d2-regression}

At $d=2$, the closed computing formula for $\bar\omega^2_d$ coincides with Zimmerman's origin-invariant bivariate statistic \cite{Zimmerman1993}. The comparison is made at the level of the explicit sample formula, avoiding dependence on how the four corner transformations are enumerated.

\begin{theorem}[$d=2$ regression]\label{thm:d2-regression}
For $n$ samples $(u_i,v_i)\in[0,1]^2$,
\[
  n\cdot\bar\omega^2_2 \;=\; \frac{1}{4n}\sum_{i,j=1}^n(1-|u_i-u_j|)(1-|v_i-v_j|) - \frac{1}{2}\sum_{i=1}^n(u_i^2-u_i-\tfrac{1}{2})(v_i^2-v_i-\tfrac{1}{2}) + \frac{n}{9},
\]
which is exactly Zimmerman's (1993, eq.~3) computing formula for $\bar\omega^2$.
\end{theorem}

\begin{proof}
The orbit-averaged factors at $d=2$ have closed forms $\mathrm{pairFactorAvg}(a,b) = \frac{1}{4}(1-|a_0-b_0|)(1-|a_1-b_1|)$ and $\mathrm{crossFactorAvg}(a) = \frac{1}{4}(a_0^2-a_0-\frac{1}{2})(a_1^2-a_1-\frac{1}{2})$. Substituting these identities into Theorem~\ref{thm:omegaBar-computing} and multiplying by $n$ gives the displayed expression, which is Zimmerman's equation~(3).
\end{proof}

\subsection{Finite-dimensional Gaussian limit}\label{sec:fidi-clt}

Under the null ($n$ iid $\mathrm{Uniform}([0,1]^d)$ samples), the normalized corner-$c$ empirical process at a fixed evaluation point $x$ is
\[
  Z_n^c(x) \;=\; \frac{1}{\sqrt{n}}\sum_{i=1}^n\bigl(\mathbf{1}\{X_i\in\mathrm{cornerLe}(c,x)\} - F^c(x)\bigr).
\]

\begin{theorem}[Scalar fidi-CLT, base corner]\label{thm:scalar-clt}
For any $x\in[0,1]^d$, $Z_n^{c_0}(x)$ converges in distribution to $N(0, F^{c_0}(x)(1-F^{c_0}(x)))$.
\end{theorem}

\begin{proof}
The summands $\mathbf{1}\{X_i\in\mathrm{cornerLe}(c_0,x)\} - F^{c_0}(x)$ are iid, bounded, mean-zero, with variance $F^{c_0}(x)(1-F^{c_0}(x))$. The classical iid CLT applies.
\end{proof}

The vector-valued, all-corners extension removes the restriction to a single fixed corner: each index $j$ carries its own corner $c_j$ and evaluation point $x_j$.

\begin{theorem}[Vector fidi-CLT, all corners]\label{thm:fidi-clt}
Let $\iota$ be a finite index set, $c : \iota\to\mathrm{Corner}(d)$ a corner assignment, and $x : \iota\to[0,1]^d$ evaluation points. The normalized centered empirical vector
\[
  \bigl(Z_n^{c_j}(x_j)\bigr)_{j\in\iota}
\]
converges jointly in distribution to a mean-zero multivariate Gaussian $N(0,S)$ with covariance
\[
  S_{j,l} \;=\; \Prob\bigl(X\in\mathrm{cornerLe}(c_j,x_j)\cap\mathrm{cornerLe}(c_l,x_l)\bigr) - F^{c_j}(x_j)\cdot F^{c_l}(x_l).
\]
\end{theorem}

\begin{proof}
By the Cram\'er--Wold device, it suffices to prove a scalar CLT for every direction $t\in\R^\iota$: the projection $\sum_j t_j Z_n^{c_j}(x_j) = n^{-1/2}\sum_i\sum_j t_j(\mathbf{1}_{ij} - F^{c_j}(x_j))$ is a sum of iid bounded mean-zero random variables (one per sample $i$), with variance $t^\top S\,t$. The classical iid CLT gives convergence to $N(0, t^\top S\,t)$ for each $t$; Cram\'er--Wold identifies the joint limit as $N(0,S)$. Positive-semidefiniteness of $S$ follows from $\Var[\sum_j t_j\mathbf{1}_j]\ge 0$.

The joint probability $\Prob(X\in\mathrm{cornerLe}(c_j,x_j)\cap\mathrm{cornerLe}(c_l,x_l))$ has a closed form per coordinate: for each $k$, the intersection of the two half-lines (one from each corner) is either a half-line (same direction) or an interval (opposite directions, possibly empty), giving a 4-branch product formula.
\end{proof}

\subsection{Existence of the limit process}\label{sec:limit-process}

Theorem~\ref{thm:fidi-clt} gives the joint Gaussian limit at any finite selection of corner--point pairs. Define
\[
  I_d:=\mathrm{Corner}(d)\times[0,1]^d
\]
and let $K_d(j,l)$ denote the covariance in Theorem~\ref{thm:fidi-clt}. The existence of a probability measure on the product process space $\R^{I_d}$ whose finite-dimensional marginals are exactly these limits follows from the Kolmogorov extension theorem.

\begin{theorem}[Limit process exists]\label{thm:limit-process}
There exists a probability measure $\mu_\infty$ on $\R^{I_d}$ whose finite-dimensional marginals are exactly the Gaussian limits of Theorem~\ref{thm:fidi-clt}, for every $d\ge 1$.
\end{theorem}

\begin{proof}
For each finite $A\subset I_d$, let $\mu_A$ be the centered Gaussian measure on $\R^A$ with covariance matrix $(K_d(j,l))_{j,l\in A}$. This family is projectively consistent because coordinate projections of Gaussian measures retain the corresponding covariance submatrices. The kernel is positive semidefinite by Theorem~\ref{thm:fidi-clt}; hence the Kolmogorov extension theorem yields a measure $\mu_\infty$ on $\R^{I_d}$ with the required marginals.
\end{proof}

\subsection{Process-level weak convergence}\label{sec:chaining}

For $j=(c,x)\in I_d$, let
\[
  f_j(y)=\mathbf 1\{y\in\mathrm{cornerLe}(c,x)\},
  \qquad
  \rho(j,l)=\bigl\|f_j-f_l\bigr\|_{L^2(P_0)},
\]
where $P_0$ is uniform measure on $[0,1]^d$. The empirical process is
\[
  Z_n(j)=\sqrt n\,(P_n-P_0)f_j.
\]

\begin{proposition}[VC structure of the corner class]\label{prop:vc-class}
The class $\mathcal F_d=\{f_j:j\in I_d\}$ is a finite union of uniformly bounded VC classes. It is therefore $P_0$-Donsker and totally bounded under $\rho$.
\end{proposition}

\begin{proof}
For a fixed corner $c$, the sets $\mathrm{cornerLe}(c,x)$ are coordinatewise lower orthants after reflecting the coordinates selected by $c$. Lower orthants form a VC class of finite dimension, and reflection preserves the VC property. Since there are only $2^d$ corners, $\mathcal F_d$ is a finite union of such classes. Its envelope is the constant function one. Restricting thresholds to rational coordinates gives a countable pointwise-dense subclass, so the standard measurability requirement is satisfied. The bounded VC-class Donsker theorem and its entropy consequences now apply \cite[Chs.~2.5--2.6]{vanderVaartWellner1996}.
\end{proof}

\begin{theorem}[Full process-level weak convergence]\label{thm:process-outer-wc}
Under iid CSR,
\[
  Z_n\rightsquigarrow G
  \qquad\text{in }\ell^\infty(I_d),
\]
in the Hoffmann--J\o rgensen outer-expectation sense, where $G$ is the centered tight Gaussian process with covariance
\[
  \Cov(G(j),G(l))
  =P_0(f_jf_l)-P_0f_j\,P_0f_l.
\]
In particular, outer expectations converge for every bounded Lipschitz functional on $\ell^\infty(I_d)$.
\end{theorem}

\begin{proof}
Proposition~\ref{prop:vc-class} shows that $\mathcal F_d$ is $P_0$-Donsker. The empirical-process central limit theorem for bounded VC classes gives the asserted weak convergence and identifies the covariance kernel. Its finite-dimensional marginals agree with Theorem~\ref{thm:fidi-clt}, hence with the process constructed in Theorem~\ref{thm:limit-process}.
\end{proof}

\begin{corollary}[Asymptotic equicontinuity]\label{thm:equicontinuity}
For every $\varepsilon>0$,
\[
  \lim_{\delta\downarrow0}\limsup_{n\to\infty}
  \Prob^*\!\left(
    \sup_{\rho(j,l)<\delta}|Z_n(j)-Z_n(l)|>\varepsilon
  \right)=0.
\]
The limit process has a version with bounded, uniformly $\rho$-continuous sample paths. Within each fixed corner,
\[
  \rho((c,x),(c,y))^2\le\sum_{k=1}^d|x_k-y_k|,
\]
so this version is continuous in the location coordinate.
\end{corollary}

\begin{proof}
Both assertions are standard consequences of the Donsker property and asymptotic equicontinuity theorem \cite[Thm.~1.5.7]{vanderVaartWellner1996}. The displayed bound follows by observing that membership in two same-corner orthants can differ only if at least one coordinate lies between the corresponding thresholds, and then applying the union bound under $P_0$.
\end{proof}

\subsection{The \texorpdfstring{CSR$\leftrightarrow$Poisson}{CSR--Poisson} bridge}\label{sec:poisson}

The working definition of CSR throughout this paper is ``$n$ iid $\mathrm{Uniform}([0,1]^d)$ points,'' which is the conditional-on-$n$ form. In the unconditional formulation, the total count is Poisson and, conditional on that count, the locations are iid uniform. This standard construction also gives the Poisson count law in every measurable sub-region \cite[Example~16.2]{vanLieshout2010}.

\begin{theorem}[Homogeneous Poisson point process on the unit cube]\label{thm:poisson-process}
There exists a probability measure $\mathrm{PPM}(\mu,d)$ on the ordered-list space
\[
  \mathrm{PointList}(d)=\coprod_{n\ge0}(\mathrm{Fin}\,n\to[0,1]^d)
\]
such that the count $N$ has law $\mathrm{Poisson}(\mu)$ and, conditional on $N=n$, the $n$ listed locations are iid $\mathrm{Uniform}([0,1]^d)$. Mapping each list to its counting measure gives the usual homogeneous Poisson point process on the unit cube.
\end{theorem}

\begin{proof}
Define $\mathrm{PPM}(\mu,d) = \sum_{n\ge 0}\mathrm{Poisson\_pmf}(\mu,n)\cdot(\mathrm{Sigma.mk}_n)_*(\mathrm{Uniform}([0,1]^d))^{\otimes n}$, i.e., weight each $n$-fiber's pushed-forward product measure by the Poisson pmf at $n$. The total mass is $\sum_n\mathrm{Poisson\_pmf}(\mu,n) = 1$. Pushing forward the count $N = \mathrm{Sigma.fst}$ gives $\mathrm{Poisson}(\mu)$, while restriction to a fixed fiber gives the stated iid product law. The final assertion follows by mapping a list $(x_i)_{i=1}^n$ to the counting measure $\sum_i\delta_{x_i}$.
\end{proof}

\begin{theorem}[Poisson thinning]\label{thm:thinning}
Let $A\subseteq[0,1]^d$ be measurable with $p=P_0(A)>0$, where $P_0$ is uniform measure on the unit cube. Restricting the process to $A$ yields a Poisson process with mean count $\mu p$ and conditional location law $P_0(\,\cdot\mid A)$. In the ordered-list representation,
\[
  (\mathrm{restrict}_A)_*\,\mathrm{PPM}(\mu,d) \;=\; \mathrm{PPM}_{\mathrm{on}}(\mu\cdot p,\,P_0(\,\cdot\mid A)),
\]
where $P_0(\,\cdot\mid A)$ is the normalized restriction of uniform measure to $A$. If $p=0$, the restricted process is empty almost surely and no conditional law on $A$ is required.
\end{theorem}

\begin{proof}
The proof assembles two halves. \emph{Count half:} for $n$ iid draws, the measure of ``$k$ of $n$ land in $A$'' is the binomial probability $\binom{n}{k}p^k(1-p)^{n-k}$ (a disjoint decomposition over which $k$-element subset survives). Summing over $n$ with Poisson weights and applying the Poisson thinning identity $\sum_n\mathrm{Poisson\_pmf}(\mu,n)\binom{n}{k}p^k(1-p)^{n-k} = \mathrm{Poisson\_pmf}(\mu p, k)$ gives the count law $\mathrm{Poisson}(\mu p)$. \emph{Value half:} for a fixed surviving-index set $S$ of size $k$, the retained values reindexed to $\mathrm{Fin}\,k$ are iid with law $P_0(\,\cdot\mid A)$ (marginalizing a product measure down to a sub-family drops the non-surviving coordinates' total mass). Reassembling over all $n$ and $S$ gives the equality of measures.
\end{proof}

\subsection{Strictly stationary mixing extension}\label{sec:mixing-extension}

The iid finite-dimensional limit extends to strictly stationary dependent observations under classical strong-mixing conditions. The argument uses the standard $\alpha$-mixing coefficient and its covariance inequality, absolute summability of the resulting covariance series, the stationary covariance-by-lag variance identity, and a scalar mixing CLT followed by the Cram\'er--Wold device. These ingredients are standard in the theory of strongly mixing sequences \cite{Ibragimov1962,Doukhan1994,Bradley2005,Rio2017}. Let $(Y_i)_{i\in\Z}$ be a strictly stationary sequence of measurable random points, each uniformly distributed on $[0,1]^d$. For a corner--point pair $(c,p)$, write
\[
  \xi_{c,p}(y)=\mathbf 1\{y\in\mathrm{cornerLe}(c,p)\}-F^c(p).
\]
Then $\xi_{c,p}$ is measurable, centered, and bounded in absolute value by one. For $r\ge1$, define the strong-mixing coefficient
\[
  \alpha(r)=\sup_{m\in\Z}
  \alpha\!\left(
    \sigma(Y_i:i\le m),
    \sigma(Y_i:i\ge m+r)
  \right).
\]
Suppose that
\begin{equation}\label{eq:mixing-rate}
  \alpha(r)\le\varphi(r),\qquad r\ge1,
\end{equation}
for a nonnegative function $\varphi:\N\to[0,\infty)$ such that
\begin{equation}\label{eq:weighted-mixing-summability}
  \sum_{r=0}^{\infty}r\,\varphi(r)<\infty.
\end{equation}
For the Gaussian limit, assume additionally that for some $\delta>0$,
\begin{equation}\label{eq:mixing-clt-summability}
  \sum_{r=1}^{\infty}\varphi(r)^{\delta/(2+\delta)}<\infty.
\end{equation}
This is a classical sufficient rate condition for the stationary strong-mixing CLT; the required $(2+\delta)$th moments hold automatically for the bounded indicator projections considered here \cite{Ibragimov1962,Doukhan1994,Rio2017}. This subsection concerns a dependent sequence with uniform one-point marginals. It does not assert that such a sequence is CSR, since CSR additionally requires independence.

\begin{theorem}[All-corners covariance summability]\label{thm:mixing-cov-summable}
For any finite collection of corner--point pairs $(c_a,p_a)_{a=1}^m$ and every $a,b\in\{1,\dots,m\}$,
\[
  \sum_{r\in\Z}
  \left|\Cov\!\left(
    \xi_{c_a,p_a}(Y_0),
    \xi_{c_b,p_b}(Y_r)
  \right)\right|<\infty.
\]
Consequently, the covariance series of every finite Cram\'er--Wold projection $X_i(t)=\sum_{a=1}^m t_a\xi_{c_a,p_a}(Y_i)$ is absolutely summable.
\end{theorem}

\begin{theorem}[Mixing variance limit]\label{thm:mixing-var-limit}
For $x\in\R^m$, define $Z_i(x)=\sum_{a=1}^m x_a\xi_{c_a,p_a}(Y_i)$ and $\gamma_x(r)=\Cov(Z_0(x),Z_r(x))$. Then
\[
  \lim_{n\to\infty}\frac{1}{n}
  \Var\!\left(\sum_{i=0}^{n-1}Z_i(x)\right)
  =\sum_{r\in\Z}\gamma_x(r)=x^TSx\ge0,
\]
where
\[
  S_{a,b}=\sum_{r\in\Z}
  \Cov\!\left(\xi_{c_a,p_a}(Y_0),\xi_{c_b,p_b}(Y_r)\right).
\]
In particular, the long-run covariance matrix $S$ is positive semidefinite.
\end{theorem}

\begin{theorem}[Finite-dimensional all-corners mixing CLT]\label{thm:mixing-joint-clt}
Under \eqref{eq:mixing-clt-summability},
\[
  \frac{1}{\sqrt n}\sum_{k=0}^{n-1}
  \left(
    \mathbf 1\{Y_k\in\mathrm{cornerLe}(c_j,p_j)\}-F^{c_j}(p_j)
  \right)_{j=1}^m
  \xrightarrow{\;\mathcal D\;}\mathcal N(0,S).
\]
The proofs of Theorems~\ref{thm:mixing-cov-summable}--\ref{thm:mixing-joint-clt} are given in Appendix~\ref{app:mixing-summability}.
\end{theorem}

\subsection{Monte Carlo validation}\label{sec:montecarlo}

The computing formula of Theorem~\ref{thm:omegaBar-computing} is directly translatable into numerical code. Table~\ref{tab:montecarlo} reports Monte Carlo null distribution percentiles for $d=1,\dots,5$ and $n\in\{30,50,60,100\}$, with $B=10{,}000$ replications.

\begin{table}[htbp]
\caption{Monte Carlo null distribution of $n\bar\omega^2_d$ under CSR ($B=10{,}000$ replications, seed 20260722).}\label{tab:montecarlo}
\begin{tabular*}{\textwidth}{@{\extracolsep\fill}lcccccc}
\toprule
$d$ & $n$ & mean & sd & p90 & p95 & p99 \\
\midrule
1 & 30  & 0.1683 & 0.1481 & 0.3500 & 0.4646 & 0.7278 \\
1 & 50  & 0.1662 & 0.1491 & 0.3424 & 0.4574 & 0.7492 \\
1 & 60  & 0.1656 & 0.1502 & 0.3409 & 0.4492 & 0.7467 \\
1 & 100 & 0.1666 & 0.1484 & 0.3460 & 0.4654 & 0.7274 \\
2 & 30  & 0.1383 & 0.0724 & 0.2296 & 0.2778 & 0.3954 \\
2 & 50  & 0.1381 & 0.0713 & 0.2317 & 0.2787 & 0.3809 \\
2 & 60  & 0.1399 & 0.0730 & 0.2349 & 0.2842 & 0.3916 \\
2 & 100 & 0.1390 & 0.0728 & 0.2326 & 0.2818 & 0.3930 \\
3 & 30  & 0.0880 & 0.0307 & 0.1291 & 0.1473 & 0.1865 \\
3 & 50  & 0.0879 & 0.0313 & 0.1289 & 0.1478 & 0.1891 \\
3 & 60  & 0.0880 & 0.0303 & 0.1276 & 0.1461 & 0.1861 \\
3 & 100 & 0.0880 & 0.0304 & 0.1284 & 0.1462 & 0.1862 \\
4 & 30  & 0.0501 & 0.0121 & 0.0657 & 0.0728 & 0.0875 \\
4 & 50  & 0.0501 & 0.0121 & 0.0660 & 0.0730 & 0.0882 \\
4 & 60  & 0.0502 & 0.0123 & 0.0665 & 0.0733 & 0.0879 \\
4 & 100 & 0.0502 & 0.0124 & 0.0665 & 0.0733 & 0.0889 \\
5 & 30  & 0.0271 & 0.0047 & 0.0333 & 0.0357 & 0.0411 \\
5 & 50  & 0.0271 & 0.0046 & 0.0333 & 0.0358 & 0.0404 \\
5 & 60  & 0.0271 & 0.0047 & 0.0334 & 0.0359 & 0.0414 \\
5 & 100 & 0.0270 & 0.0047 & 0.0331 & 0.0357 & 0.0411 \\
\botrule
\end{tabular*}
\end{table}

The null mean is available exactly for every sample size. Indeed,
\[
  \E[n\bar\omega_d^2]
  = \frac{1}{2^d}\sum_c\int_{[0,1]^d}F^c(x)\bigl(1-F^c(x)\bigr)\,dx
  = \frac{1}{2^d}-\frac{1}{3^d},
\]
because $\int F^c=2^{-d}$ and $\int(F^c)^2=3^{-d}$ for every corner. Thus the exact means for $d=1,\dots,5$ are approximately $0.1667$, $0.1389$, $0.0880$, $0.0502$, and $0.0271$, respectively, in agreement with the simulated values. The stability of the means across $n$ is therefore an exact finite-sample property, not evidence by itself of distributional convergence. The reported quantiles provide simulation-based critical values at the listed $(d,n)$ combinations; their Monte Carlo uncertainty is largest in the extreme tail.

\subsection{Finite-sample size}\label{sec:size}

To assess the Monte Carlo null calibration presented in Table~\ref{tab:montecarlo}, we perform a finite-sample size study. For each dimension $d \in \{1,\dots,5\}$ and sample size $n \in \{30,50,60,100\}$, we generate $B = 5{,}000$ additional independent realizations under CSR. For each realization, the statistic $n\bar\omega^2_d$ is computed and compared against the corresponding simulated $p_{95}$ critical value from Table~\ref{tab:montecarlo}.

Table~\ref{tab:size} reports the empirical rejection rates at nominal level $\alpha=0.05$, together with approximate $95\%$ binomial margins of error. The observed rates are compatible with $0.05$ at the resolution of $5{,}000$ replications. This supports the internal calibration of the reported experiment but does not replace an independently reproducible implementation.

\begin{table}[htbp]
\caption{Empirical size of the test at nominal $\alpha = 0.05$ over $B = 5{,}000$ additional CSR replications.}\label{tab:size}
\begin{tabular*}{\textwidth}{@{\extracolsep\fill}lcccc}
\toprule
$d$ & $n=30$ & $n=50$ & $n=60$ & $n=100$ \\
\midrule
1 & $0.0532 \pm 0.0062$ & $0.0506 \pm 0.0061$ & $0.0544 \pm 0.0063$ & $0.0546 \pm 0.0063$ \\
2 & $0.0460 \pm 0.0058$ & $0.0460 \pm 0.0058$ & $0.0458 \pm 0.0058$ & $0.0448 \pm 0.0057$ \\
3 & $0.0480 \pm 0.0059$ & $0.0484 \pm 0.0059$ & $0.0482 \pm 0.0059$ & $0.0538 \pm 0.0063$ \\
4 & $0.0536 \pm 0.0062$ & $0.0512 \pm 0.0061$ & $0.0494 \pm 0.0060$ & $0.0528 \pm 0.0062$ \\
5 & $0.0552 \pm 0.0063$ & $0.0432 \pm 0.0056$ & $0.0498 \pm 0.0060$ & $0.0522 \pm 0.0062$ \\
\bottomrule
\end{tabular*}
\end{table}

\subsection{Consistency and power against specified alternatives}\label{sec:consistency}

We now establish consistency under a stated class of fixed iid alternatives and examine finite-sample rejection rates for selected alternatives.

\begin{proposition}[Consistency of the multivariate CvM test]\label{prop:consistency}
Let $P$ be a probability distribution on $[0,1]^d$, with lower-orthant distribution function $F_P$, and let $F_0(x)=\prod_{k=1}^d x_k$ be the uniform distribution function. If $F_P\ne F_0$ on a set of positive Lebesgue measure, then under iid sampling from $P$,
\[
  n\bar\omega^2_d \xrightarrow{P} \infty \quad \text{as } n\to\infty.
\]
\end{proposition}

\begin{proof}
By \eqref{eq:omegaBarSqD}, the statistic is an average over the $2^d$-corner orbit:
\[
  \bar\omega^2_d \;=\; \frac{1}{2^d}\sum_{c\in\mathrm{Corner}(d)}\int_{[0,1]^d}\bigl(F_n^c(x) - F_0^c(x)\bigr)^2\,dx.
\]
For each corner $c$, let $F_P^c(x)=P(\mathrm{cornerLe}(c,x))$. The multivariate Glivenko--Cantelli theorem gives uniform almost-sure convergence of $F_n^c$ to $F_P^c$. Since all functions involved are bounded by one,
\[
  \int_{[0,1]^d}\bigl(F_n^c(x) - F_0^c(x)\bigr)^2\,dx \xrightarrow{a.s.} \int_{[0,1]^d}\bigl(F_P^c(x) - F_0^c(x)\bigr)^2\,dx.
\]
Summing over corners, we obtain $\bar\omega^2_d \xrightarrow{a.s.} \bar\omega^2_{d, \infty}$, where
\[
  \bar\omega^2_{d,\infty} \;=\; \frac{1}{2^d}\sum_{c\in\mathrm{Corner}(d)}\int_{[0,1]^d}\bigl(F_P^c(x) - F_0^c(x)\bigr)^2\,dx.
\]
The base-corner term is already strictly positive because $F_P\ne F_0$ on a set of positive measure. Hence $\bar\omega^2_{d,\infty}>0$, and $n\bar\omega^2_d=n\bar\omega^2_{d,\infty}+o_{a.s.}(n)\to\infty$ almost surely.
\end{proof}

To examine finite-sample behavior, we simulate iid coordinate-wise Beta alternatives $\mathrm{Beta}(\alpha_0,\alpha_0)^d$ for $d\in\{1,2,3\}$ and $n\in\{30,50,100\}$, using $B=5{,}000$ replications. The case $\alpha_0=0.5$ is boundary- and corner-concentrated, while $\alpha_0=2$ is center-concentrated. Both are first-order inhomogeneous iid alternatives; neither represents interaction-driven clustering or inhibition.

Table~\ref{tab:power-beta} reports empirical rejection rates at $\alpha=0.05$. In these configurations the rejection rate increases with $n$, as predicted by Proposition~\ref{prop:consistency}. It also increases with $d$ because the chosen alternatives impose the same marginal departure independently in every coordinate; this should not be interpreted as a universal monotonicity-in-dimension result.

\begin{table}[htbp]
\caption{Empirical power of $n\bar\omega^2_d$ at $\alpha = 0.05$ against $\mathrm{Beta}(\alpha_0, \alpha_0)^d$ alternatives ($B = 5{,}000$ replications).}\label{tab:power-beta}
\begin{tabular*}{\textwidth}{@{\extracolsep\fill}lcccc}
\toprule
$d$ & $\alpha_0$ & $n=30$ & $n=50$ & $n=100$ \\
\midrule
1 & 0.5 & 0.2666 & 0.4606 & 0.8410 \\
1 & 2.0 & 0.0882 & 0.2512 & 0.7198 \\
2 & 0.5 & 0.4768 & 0.7582 & 0.9906 \\
2 & 2.0 & 0.2568 & 0.6256 & 0.9868 \\
3 & 0.5 & 0.6406 & 0.9178 & 1.0000 \\
3 & 2.0 & 0.4794 & 0.8738 & 0.9998 \\
\bottomrule
\end{tabular*}
\end{table}

\subsection{Power comparison with competing tests}\label{sec:power-comparison}

We compare the statistical power of the proposed origin-invariant bivariate CvM statistic $\bar\omega^2_2$ against three standard competitors in $\mathbb R^2$:
\begin{enumerate}
  \item The symmetric discrepancy $D_2^{(S)}$ of Chiu and Liu \cite{ChiuLiu2009}, which averages the CvM corner functional over only the $2^{d-1} = 2$ even-parity corners.
  \item The classical Clark--Evans nearest-neighbor index $R$ \cite{ClarkEvans1954}, using $|R - 1|$ as a two-sided test statistic to capture both clumping and regularity.
  \item Ripley's centered $L$-function maximum deviation statistic $L_{\max}=\sup_{0\le r\le0.25}|\widehat L(r)-r|$ \cite{Diggle2003}.
\end{enumerate}
All tests are calibrated under CSR using $B_{\mathrm{calib}}=3{,}000$ simulations to estimate upper $0.95$ critical quantiles. Rejection rates are then estimated from $B=1{,}000$ replications for $n=50$ and $n=100$ under four reported simulation designs: independent coordinate-wise Beta$(0.5,0.5)$ points; a Mat\'ern cluster process ($\kappa=10$, scale $0.08$, $\mu=10$); a hard-core process ($r=0.045$); and a Strauss process ($r=0.08$, $\gamma=0.1$, using 30 simulation sweeps). Because the code, initialization, edge corrections, and count-conditioning rules are not included in the present manuscript package, these comparisons should be read as illustrative rather than definitive benchmarks.

Table~\ref{tab:comparison} shows the empirical rejection rates. In the two inhomogeneous or clustered designs shown, $\bar\omega^2_2$ has higher estimated rejection rates than the Chiu--Liu even-parity average (for example, $0.7440$ versus $0.6280$ for the Beta design and $0.9110$ versus $0.8480$ for the Mat\'ern design at $n=50$). The experiment does not establish uniform power dominance or robustness over a broader alternative class.

For the reported inhibition designs, the distance-based statistics have much larger rejection rates. The upper-tailed CvM calibration is poorly aligned with these alternatives because the observed CvM values tend to fall in the lower tail. A lower-tailed or two-sided calibration is therefore required if regularity is an intended alternative. This limitation is consistent with Zimmerman's observation that the EDF-based statistic is primarily effective for heterogeneous alternatives and can be weak against regular or aggregated interaction alternatives \cite{Zimmerman1993}.

\begin{table}[htbp]
\caption{Power comparison for $d=2$ at significance level $\alpha = 0.05$ ($B = 1{,}000$ replications).}\label{tab:comparison}
\begin{tabular*}{\textwidth}{@{\extracolsep\fill}llcc}
\toprule
Alternative & Test & $n=50$ & $n=100$ \\
\midrule
\multicolumn{4}{l}{\textit{Symmetric Beta clustering ($\mathrm{Beta}(0.5,0.5)^2$)}} \\
  & Our $\bar\omega^2_2$  & 0.7440 & 0.9860 \\
  & Chiu--Liu $D_2^{(S)}$ & 0.6280 & 0.9780 \\
  & Clark--Evans $|R-1|$  & 0.0540 & 0.0240 \\
  & Ripley $L_{\max}$     & 0.8100 & 0.9940 \\
\midrule
\multicolumn{4}{l}{\textit{Mat\'ern cluster process}} \\
  & Our $\bar\omega^2_2$  & 0.9110 & 0.9660 \\
  & Chiu--Liu $D_2^{(S)}$ & 0.8480 & 0.9330 \\
  & Clark--Evans $|R-1|$  & 0.9990 & 1.0000 \\
  & Ripley $L_{\max}$     & 0.9890 & 0.9800 \\
\midrule
\multicolumn{4}{l}{\textit{Hard-core process}} \\
  & Our $\bar\omega^2_2$  & 0.0110 & 0.0000 \\
  & Chiu--Liu $D_2^{(S)}$ & 0.0130 & 0.0000 \\
  & Clark--Evans $|R-1|$  & 0.6750 & 1.0000 \\
  & Ripley $L_{\max}$     & 0.2720 & 1.0000 \\
\midrule
\multicolumn{4}{l}{\textit{Strauss process}} \\
  & Our $\bar\omega^2_2$  & 0.0010 & 0.0000 \\
  & Chiu--Liu $D_2^{(S)}$ & 0.0020 & 0.0000 \\
  & Clark--Evans $|R-1|$  & 0.9970 & 1.0000 \\
  & Ripley $L_{\max}$     & 0.8240 & 0.9400 \\
\bottomrule
\end{tabular*}
\end{table}

\section{Real-data illustration}\label{sec:application}

We illustrate the statistic on three bivariate datasets distributed with the \texttt{spatstat.data} package \cite{BaddeleyRubakTurner2015}:
\begin{enumerate}
  \item \textbf{Swedish pines} ($n=71$): pine saplings observed in a rectangular $9.6\,\mathrm{m}\times10\,\mathrm{m}$ forest plot;
  \item \textbf{Redwood seedlings} ($n=62$): Ripley's unit-square subset of the California redwood seedling and sapling data, commonly used as a clustered example;
  \item \textbf{Biological cells} ($n=42$): a pattern of biological cell centers commonly used as an example of regular spacing.
\end{enumerate}
Each coordinate is affinely rescaled to $[0,1]$. We compute $n\bar\omega^2_2$ and estimate upper- and lower-tail probabilities from $B=50{,}000$ CSR simulations matched to the observed count. The two tails are descriptive diagnostics: a large value indicates a large EDF departure, whereas an unusually small value may occur for highly regular configurations. Monte Carlo $p$-values are reported using the correction $(b+1)/(B+1)$ and therefore are never zero \cite{PhipsonSmyth2010}.

Figure~\ref{fig:real-data} displays the rescaled patterns, and Table~\ref{tab:real-data} summarizes the numerical results. The biological-cells statistic lies below every simulated null value, giving the corrected lower-tail value $1/50{,}001\approx0.00002$. The redwood and Swedish-pines values are not significant in the displayed one-sided tails. These examples illustrate differing sensitivity of the statistic but do not establish that it is an omnibus diagnostic for all forms of clustering or inhibition.

\begin{figure}[htbp]
\centering
\includegraphics[width=\textwidth]{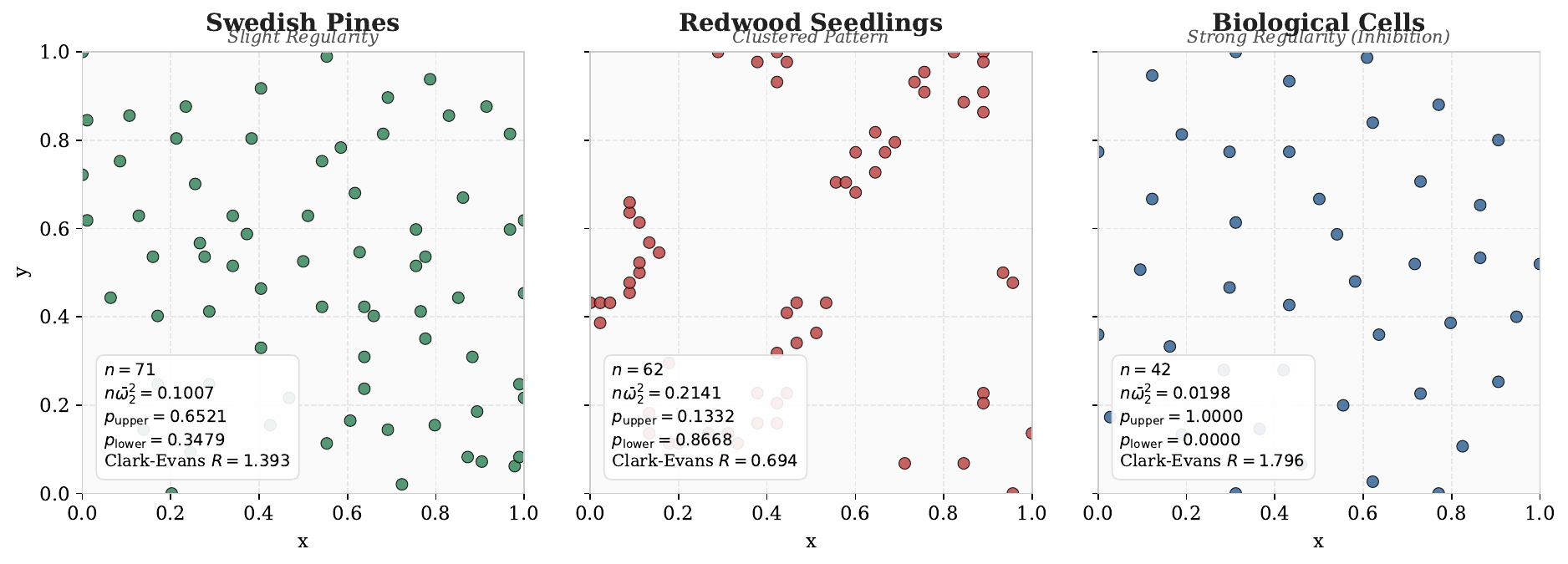}
\caption{The Swedish-pines, redwood-seedlings, and biological-cells point patterns after affine rescaling to $[0,1]^2$. Insets report the observed $n\bar\omega_2^2$, Monte Carlo tail probabilities, and Clark--Evans index. Probabilities in the figure are rounded to four decimal places; Table~\ref{tab:real-data} reports the corrected biological-cells lower-tail probability to greater precision.}
\label{fig:real-data}
\end{figure}

\begin{table}[htbp]
\caption{Observed statistics and Monte Carlo $p$-values ($B = 50{,}000$ null simulations) for the spatial datasets.}\label{tab:real-data}
\begin{tabular*}{\textwidth}{@{\extracolsep\fill}lccccc}
\toprule
Dataset & $n$ & $n\bar\omega^2_2$ & $p_{\mathrm{upper}}$ & $p_{\mathrm{lower}}$ & Clark--Evans $R$ \\
\midrule
Swedish pines      & 71 & 0.1007 & 0.6521 & 0.3479 & 1.393 \\
Redwood seedlings  & 62 & 0.2141 & 0.1332 & 0.8668 & 0.694 \\
Biological cells   & 42 & 0.0198 & 1.0000 & 0.00002 & 1.796 \\
\bottomrule
\end{tabular*}
\end{table}

\section{Asymptotic null law}\label{sec:null-law}

Let $I_d=\mathrm{Corner}(d)\times[0,1]^d$ and equip it with the product measure
\[
  \nu_d=2^{-d}\sum_{c\in\mathrm{Corner}(d)}\delta_c\otimes dx.
\]
Then $n\bar\omega_d^2=Q(Z_n)$, where $Q(g)=\int_{I_d}g(j)^2\,d\nu_d(j)$. The map $Q$ is continuous under the supremum norm on bounded functions. The process-level weak convergence from Theorem~\ref{thm:process-outer-wc} establishes $Z_n\rightsquigarrow G$ in $\ell^\infty(I_d)$ in the sense of van der Vaart and Wellner \cite[\S1.9]{vanderVaartWellner1996}. Applying the continuous mapping theorem to $Q$ yields $n\bar\omega_d^2\Rightarrow Q(G)$.

\begin{theorem}[Asymptotic null law]\label{thm:null-dist}
Under CSR,
\[
  n\bar\omega^2_d \xrightarrow{d} \sum_{k=1}^\infty \lambda_k Z_k^2,
\]
where $\{Z_k\}$ are iid $N(0,1)$ and $\{\lambda_k\}$ are the eigenvalues of the covariance operator $T_d$ on $\mathcal H_d=L^2(I_d,\nu_d)$,
\[
  (T_df)(j)=\int_{I_d}K_d(j,l)f(l)\,d\nu_d(l).
\]
For $d=1$ this specializes to the classical $\sum_k Z_k^2/(k^2\pi^2)$ \cite{KacSiegert1947,AndersonDarling1952,DurbinKnott1972}; for $d=2$ it is the origin-invariant bivariate limit studied by Zimmerman \cite{Zimmerman1994}.
\end{theorem}

\begin{proof}
By Theorem~\ref{thm:process-outer-wc} and the continuous mapping theorem applied to $Q$, $n\bar\omega_d^2\Rightarrow Q(G)$. The covariance-operator properties and trace identity are proved in Theorems~\ref{thm:op-self-adjoint} and~\ref{thm:nonneg-eigen}. The Karhunen--Lo\`eve expansion of the centered Gaussian element $G\in\mathcal H_d$ therefore gives $Q(G)=\sum_k\lambda_kZ_k^2$ \cite{GineNickl2016}.
\end{proof}

\section{Discussion}\label{sec:discussion}

\subsection{Summary of established results}

The finite-sample construction is complete. The per-corner and orbit-averaged $O(n^2)$ formulas are given by Theorems~\ref{thm:computing} and~\ref{thm:omegaBar-computing}, while Theorems~\ref{thm:d1-regression} and~\ref{thm:d2-regression} establish the exact reductions to the classical univariate statistic and Zimmerman's bivariate statistic. The Poisson formulation of CSR and its restriction law are supplied by Theorems~\ref{thm:poisson-process} and~\ref{thm:thinning}.

Under iid CSR, Theorem~\ref{thm:fidi-clt} gives the joint all-corners finite-dimensional Gaussian limit and Theorem~\ref{thm:limit-process} constructs the corresponding Gaussian process. Proposition~\ref{prop:vc-class} identifies the bounded VC structure of the index class, yielding full process-level weak convergence and asymptotic equicontinuity in Theorems~\ref{thm:process-outer-wc} and~\ref{thm:equicontinuity}. The continuous mapping argument then gives the quadratic limit in Theorem~\ref{thm:null-dist}.

The dependent-observation extension is stated in Section~\ref{sec:mixing-extension}, with proofs collected in Appendix~\ref{app:mixing-summability}. Under strict stationarity, uniform marginals, and the stated $\alpha$-mixing conditions, Theorem~\ref{thm:mixing-cov-summable} proves absolute cross-covariance summability, Theorem~\ref{thm:mixing-var-limit} identifies the positive semidefinite long-run covariance matrix, and Theorem~\ref{thm:mixing-joint-clt} gives the all-corners finite-dimensional Gaussian limit by the Cram\'er--Wold device.

Appendix~\ref{app:spectral} proves that the covariance operator is compact, positive, self-adjoint, and trace class (Lemma~\ref{lem:kernel-symm} and Theorem~\ref{thm:op-self-adjoint}), with nonnegative eigenvalues summing to $2^{-d}-3^{-d}$ (Theorem~\ref{thm:nonneg-eigen}). Appendix~\ref{app:local-power} records mean-shift identities for local density perturbations but explicitly stops short of a full local-power theorem.

\subsection{Post-hoc localization of EDF discrepancy following rejection of CSR}\label{sec:posthoc}

The statistic $n\bar\omega_d^2$ is a global discrepancy measure. Rejection provides evidence against the fixed-count CSR null but does not identify the form of the departure. A diagnostic decomposition can be obtained by partitioning the EDF evaluation domain into fixed measurable sets $A_1,\dots,A_m\subseteq[0,1]^d$ and defining
\begin{equation}\label{eq:regional-contribution}
  D_j
  =n\,2^{-d}\sum_{c\in\mathrm{Corner}(d)}
    \int_{A_j}\bigl(F_n^c(x)-F^c(x)\bigr)^2\,dx,
  \qquad j=1,\dots,m.
\end{equation}
These quantities are non-negative and satisfy $\sum_{j=1}^m D_j=n\bar\omega_d^2$ whenever the sets form a partition up to Lebesgue-null sets. This is an exact decomposition over threshold locations $x$ at which the corner EDFs are evaluated. It is not an additive decomposition of the observed points by physical subregion: because each $F_n^c(x)$ is cumulative over a corner-oriented orthant, a large $D_j$ need not imply that the generating departure is confined to $A_j$.

For a prespecified partition, complete-null multiplicity can be addressed using joint CSR simulation. For replication $b$, let $D_j^{(b)}$ be the regional contribution and set $M^{(b)}=\max_{1\le k\le m}D_k^{(b)}$. Following the single-step resampling max-$T$ principle of Westfall and Young \cite{WestfallYoung1993}, define
\begin{equation}\label{eq:maxT-adjusted}
  p_j^{\mathrm{adj}}
  =\frac{1+\#\{b:M^{(b)}\ge D_j^{\mathrm{obs}}\}}{B+1}.
\end{equation}
The correction by $B+1$ is the standard finite Monte Carlo correction \cite{PhipsonSmyth2010}. Under the complete CSR null, the observed pattern and the $B$ simulated patterns are exchangeable. Moreover, $\min_jp_j^{\mathrm{adj}}$ is exactly the Monte Carlo $p$-value obtained by comparing the observed maximum $M^{\mathrm{obs}}=\max_jD_j^{\mathrm{obs}}$ with $M^{(1)},\dots,M^{(B)}$. Consequently, rejecting any region when $p_j^{\mathrm{adj}}\le\alpha$ controls the unconditional familywise error rate under the complete null, conservatively in the presence of ties. This is weak, complete-null control; it is not a conditional error guarantee given that the global test has rejected.

Strong familywise-error control when only a subset of regional null hypotheses is true does not follow from complete-null simulation alone. Classical Westfall--Young strong-control arguments invoke subset pivotality, and alternative step-down methods require their own finite- or asymptotic resampling conditions \cite{WestfallYoung1993,RomanoWolf2005,WestfallTroendle2008}. Those conditions have not been established for the cumulative regional statistics in \eqref{eq:regional-contribution}. A step-down max-$T$ variant may be less conservative, but it should not be presented as strongly controlling until the relevant assumptions are verified. If regions have substantially different volumes or null variances, studentized statistics may also be preferable to the raw $D_j$ values, with the studentization included in every simulation.

The partition must be fixed before examining the observed pattern; data-dependent selection would require an additional selective-inference adjustment. Regular grids or prespecified multiscale partitions are natural choices for \eqref{eq:regional-contribution}. Coordinate projections, corner-symmetric contrasts, Karhunen--Lo\`eve scores, and pairwise two-sample comparisons are different diagnostic families rather than instances of the regional decomposition, and each would require its own statistic and joint null calibration.

Individual corner statistics should not be interpreted as independent post-hoc comparisons: they share the same observations, are generally dependent, and are not individually origin-invariant. The proposed regional construction is therefore a prespecified diagnostic with complete-null max-$T$ calibration, not a fully developed localization test. Strong error control, power, partition choice, studentization, and step-down implementation require separate theoretical and simulation study before routine use.

\subsection{Relationship to the literature}

The full $2^d$-corner average extends Zimmerman's origin-invariance principle \cite{Zimmerman1993} to arbitrary dimension, and Theorem~\ref{thm:d2-regression} verifies agreement of the resulting bivariate computing formula with Zimmerman's equation~(3). EDF-based multivariate tests also have a broader lineage in distribution-free goodness-of-fit and independence testing \cite{BlumKieferRosenblatt1961,ChiuLiu2009}. In spatial statistics, CSR is commonly assessed using nearest-neighbor and second-order summaries such as the Clark--Evans index and Ripley's $K$ or $L$ functions \cite{ClarkEvans1954,Ripley1977,Diggle2003,BaddeleyRubakTurner2015}. These methods target different aspects of departure from CSR; Zimmerman explicitly reported that his EDF-based statistic was strongest for heterogeneous alternatives and weaker for regular or aggregated interaction alternatives \cite{Zimmerman1993}.

The process limit is a bounded-VC-class empirical-process CLT in the sense of \cite{vanderVaartWellner1996,GineNickl2016}, specialized to the finite union of corner-oriented orthant classes. The Poisson thinning bridge (Theorems~\ref{thm:poisson-process}--\ref{thm:thinning}) gives the restriction law described in \cite[Example~16.2]{vanLieshout2010}.

The strictly stationary extension in Section~\ref{sec:mixing-extension} follows the classical strong-mixing route: covariance control through $\alpha$-mixing inequalities \cite{Doukhan1994,Bradley2005,Rio2017}, identification of the long-run variance from absolute covariance summability, and the stationary scalar mixing CLT of Ibragimov type \cite{Ibragimov1962,Doukhan1994,Rio2017}, assembled across corner--point coordinates by Cram\'er--Wold.

\subsection{Limitations and future work}

Several mathematical and computational limitations remain for future work:

\begin{enumerate}
\item \textbf{Dependent-process tightness.} Theorems~\ref{thm:mixing-var-limit} and~\ref{thm:mixing-joint-clt} establish the long-run covariance limit and finite-dimensional all-corners Gaussian convergence under strict stationarity and $\alpha$-mixing. Extending Theorem~\ref{thm:equicontinuity} to this dependent setting is needed for a full process-level mixing limit.
\item \textbf{Explicit eigenstructure.} Theorems~\ref{thm:op-self-adjoint} and~\ref{thm:nonneg-eigen} establish the operator properties and trace identity, but this paper does not derive closed-form eigenfunctions or eigenvalues for the full corner--location operator beyond the classical $d=1$ specialization.
\item \textbf{Reproducible numerical study.} Table~\ref{tab:montecarlo} uses $B=10{,}000$ null replications, while the size, power, and comparison tables use the replication counts stated in their captions. A public release of the simulation code, random-number seeds, generated raw summaries, and precise point-process simulation rules is needed for independent reproduction.
\item \textbf{Local alternatives.} Appendix~\ref{app:local-power} proves mean-shift identities only. A genuine local-power result requires triangular-array process convergence and the corresponding noncentral quadratic limit.
\end{enumerate}

\section{Conclusion}\label{sec:conclusion}

We have constructed a general-$d$ origin-invariant Cram\'er--von Mises statistic $\bar\omega^2_d$ and supplied a closed-form $O(n^2)$ formula for simulation-calibrated testing. The statistic reduces to the classical $d=1$ CvM statistic, and its $d=2$ computing formula agrees with Zimmerman's statistic. Under iid CSR, the full corner-indexed empirical process has a Gaussian weak limit with an explicit cross-corner covariance kernel; the statistic consequently converges to the quadratic Gaussian law governed by a positive trace-class covariance operator. The Poisson construction connects fixed-count CSR with the homogeneous Poisson formulation. Under strict stationarity and the stated $\alpha$-mixing assumptions, we additionally obtain the long-run covariance matrix and finite-dimensional all-corners Gaussian convergence. The numerical studies illustrate calibration and sensitivity, particularly to first-order inhomogeneity, but do not establish omnibus superiority over methods designed for interaction-driven clustering or inhibition.

\backmatter

\section*{Acknowledgments}
The author thanks Dr. Welfredo Patungan of the UP School of Statistics for agreeing to serve as an adviser for this paper. The author also acknowledges Yuanhe Zhang, Jason D. Lee, and Fanghui Liu \cite{zhang2026aislt}, whose empirical-process proof architecture for formal statistical learning theory informed portions of the proof strategy used in this work.

\section*{Statements and Declarations}

\begin{itemize}
\item \textbf{Funding:} No funds, grants, or other support were received for conducting this study or preparing this manuscript.
\item \textbf{Competing interests:} The author has no relevant financial or non-financial interests to disclose.
\item \textbf{Data availability:} The Monte Carlo simulation code and raw results are available from the author on request.
\item \textbf{Author contribution:} M.M. performed the conceptualization, methodology, formal analysis, software development, numerical investigation, and preparation and revision of the manuscript.
\item \textbf{Ethics approval:} Not applicable. This study did not involve human participants, animals, or identifiable personal data.
\item \textbf{Consent to participate and publish:} Not applicable.
\item \textbf{Use of generative AI:} Prism and Claude Code were used to assist with LaTeX editing, prose revision, and manuscript diagnostic checking. Lean 4 was used for formal proof verification. The author reviewed all AI-assisted outputs, checked their mathematical content, and accepts full responsibility for the final manuscript.
\end{itemize}

\setcounter{maintablecount}{\value{table}}
\newpage
\begin{appendices}

\setcounter{table}{\value{maintablecount}}
\renewcommand{\thetable}{\arabic{table}}
\renewcommand{\theHtable}{\arabic{table}}
\renewcommand{\theHequation}{\Alph{section}.\arabic{equation}}

\section{Notation summary}\label{sec:notation}

\begin{table}[htbp]
\caption{Notation summary.}\label{tab:notation}
\begin{tabular*}{\textwidth}{@{\extracolsep\fill}ll}
\toprule
Symbol & Meaning \\
\midrule
$d$ & dimension \\
$c$ & corner of $[0,1]^d$ (function $\{0,\dots,d-1\}\to\{\mathrm{true},\mathrm{false}\}$) \\
$c_0$ & base corner (all-false) \\
$\mathrm{cornerLe}(c,x)$ & corner-$c$ lower set of $x$ \\
$F_n^c(x)$ & corner-$c$ empirical CDF \\
$F^c(x)$ & corner-$c$ theoretical CDF under $\mathrm{Uniform}([0,1]^d)$ \\
$\omega^2_c$ & corner-$c$ CvM functional \\
$\bar\omega^2_d$ & $2^d$-corner orbit average (the statistic) \\
$\mathrm{pairFactor}_c(a,b)$ & per-corner pair factor (Definition~\ref{def:factors}) \\
$\mathrm{crossFactor}_c(a)$ & per-corner cross factor (Definition~\ref{def:factors}) \\
$Z_n^c(x)$ & normalized corner-$c$ empirical process at $x$ \\
$I_d$ & corner--location index set $\mathrm{Corner}(d)\times[0,1]^d$ \\
$K_d(j,l)$ & covariance kernel of the limit process \\
$\rho(j,l)$ & $L^2(P_0)$ semimetric on $I_d$ \\
$\mathrm{PPM}(\mu,d)$ & Poisson point process measure on $[0,1]^d$ \\
\botrule
\end{tabular*}
\end{table}

\section{Long-run covariance summability, variance limits, and joint mixing CLT}\label{app:mixing-summability}

This appendix proves Theorems~\ref{thm:mixing-cov-summable}--\ref{thm:mixing-joint-clt}. The notation and strict-stationarity assumptions are those of Section~\ref{sec:mixing-extension}.

\begin{lemma}[Pairwise covariance bound]\label{lem:mixing-cov-bound}
For any corners $c_1,c_2$, points $p_1,p_2\in[0,1]^d$, and distinct indices $i,j$,
\[
  \left|\Cov\!\left(\xi_{c_1,p_1}(Y_i),\xi_{c_2,p_2}(Y_j)\right)\right|
  \le 4\varphi(|i-j|).
\]
\end{lemma}

\begin{proof}
Assume without loss of generality that $i<j$. The first variable is measurable with respect to $\sigma(Y_k:k\le i)$ and the second with respect to $\sigma(Y_k:k\ge j)$; both are bounded in absolute value by one. The standard strong-mixing covariance inequality for variables bounded by constants $C_1$ and $C_2$ gives \cite{Doukhan1994,Bradley2005,Rio2017}
\[
  |\Cov(X,Z)|\le4C_1C_2\alpha(j-i).
\]
Taking $C_1=C_2=1$ and applying \eqref{eq:mixing-rate} with lag $j-i$ proves the result.
\end{proof}

\begin{lemma}[Shifted rate summability]\label{lem:mixing-rate-summable}
Assumption~\eqref{eq:weighted-mixing-summability} implies
\[
  \sum_{r=0}^{\infty}\varphi(r+1)<\infty.
\]
\end{lemma}

\begin{proof}
For every $r\ge0$, nonnegativity gives
\[
  \varphi(r+1)\le(r+1)\varphi(r+1).
\]
The series formed by the right-hand side is the tail of the convergent series in \eqref{eq:weighted-mixing-summability}; comparison completes the proof.
\end{proof}

\begin{proof}[Proof of Theorem~\ref{thm:mixing-cov-summable}]
For positive lags, Lemma~\ref{lem:mixing-cov-bound} bounds the lag-$r$ covariance by $4\varphi(r)$, and Lemma~\ref{lem:mixing-rate-summable} gives summability. Negative lags follow by strict stationarity and interchange of the two variables. Hence every entry of the all-corners covariance array is absolutely summable. Moreover, $|X_i(t)|\le\|t\|_1$, so the same covariance inequality gives
\[
  |\Cov(X_i(t),X_j(t))|
  \le4\|t\|_1^2\varphi(|i-j|),
\]
and its covariance series is absolutely summable.
\end{proof}

\begin{proof}[Proof of Theorem~\ref{thm:mixing-var-limit}]
Strict stationarity makes each covariance depend only on the lag, so
\[
  \frac{1}{n}\Var\!\left(\sum_{i=0}^{n-1}Z_i(x)\right)
  =\gamma_x(0)+\sum_{r=1}^{n-1}\left(1-\frac{r}{n}\right)
    \bigl(\gamma_x(r)+\gamma_x(-r)\bigr).
\]
Theorem~\ref{thm:mixing-cov-summable} makes the two-sided covariance series absolutely summable. Dominated convergence therefore yields the limit $\sum_{r\in\Z}\gamma_x(r)=x^TSx$. Since every normalized finite-sample variance is nonnegative, its limit satisfies $x^TSx\ge0$ for every $x\in\R^m$, and hence $S$ is positive semidefinite.
\end{proof}

\begin{proof}[Proof of Theorem~\ref{thm:mixing-joint-clt}]
By Theorem~\ref{thm:mixing-var-limit}, every one-dimensional linear combination has limiting normalized variance $t^TSt$. Each projection is bounded and therefore has a finite $(2+\delta)$th moment. Under \eqref{eq:mixing-clt-summability}, the stationary strong-mixing CLT gives convergence in direction $t$ to $\mathcal N(0,t^TSt)$ \cite{Ibragimov1962,Doukhan1994,Rio2017}. The Cram\'er--Wold device then assembles these directional limits into the joint $m$-dimensional Gaussian limit $\mathcal N(0,S)$.
\end{proof}

\section{Spectral analysis of the limiting covariance operator}\label{app:spectral}

This appendix records the operator properties needed for Theorem~\ref{thm:null-dist}. Put $I_d=\mathrm{Corner}(d)\times[0,1]^d$ and
\[
  \nu_d=2^{-d}\sum_{c\in\mathrm{Corner}(d)}\delta_c\otimes dx.
\]
For $j=(c,p)$, write $A_j=\mathrm{cornerLe}(c,p)$ and $m(j)=P_0(A_j)=F^c(p)$. The covariance kernel is
\[
  K_d(j,l)=P_0(A_j\cap A_l)-m(j)m(l).
\]
If $l=(c',q)$, the intersection probability has the coordinate product form
\[
  P_0(A_j\cap A_l)=\prod_{k=1}^d H_k(c,c';p,q),
\]
where
\[
  H_k(c,c';p,q)=
  \begin{cases}
    \min(p_k,q_k), & c_k=c'_k=0,\\
    1-\max(p_k,q_k), & c_k=c'_k=1,\\
    (p_k-q_k)_+, & c_k=0,\ c'_k=1,\\
    (q_k-p_k)_+, & c_k=1,\ c'_k=0.
  \end{cases}
\]

\begin{lemma}[Covariance-kernel properties]\label{lem:kernel-symm}
The kernel $K_d$ is measurable, symmetric, bounded, and positive semidefinite: for every finite set $j_1,\dots,j_m\in I_d$ and $a_1,\dots,a_m\in\R$,
\[
  \sum_{r,s=1}^m a_ra_sK_d(j_r,j_s)\ge0.
\]
\end{lemma}

\begin{proof}
The intersection formula gives measurability and boundedness, while symmetry follows from $A_j\cap A_l=A_l\cap A_j$. For positive semidefiniteness,
\[
  \sum_{r,s}a_ra_sK_d(j_r,j_s)
  =\Var_{P_0}\!\left(\sum_{r=1}^m a_r
    \bigl(\mathbf1_{A_{j_r}}-P_0(A_{j_r})\bigr)\right)\ge0.
\]
\end{proof}

Define $T_d:L^2(I_d,\nu_d)\to L^2(I_d,\nu_d)$ by
\[
  (T_df)(j)=\int_{I_d}K_d(j,l)f(l)\,d\nu_d(l).
\]

\begin{theorem}[Covariance-operator properties]\label{thm:op-self-adjoint}
The operator $T_d$ is bounded, compact, self-adjoint, positive, and trace class.
\end{theorem}

\begin{proof}
Since $K_d$ is bounded and $\nu_d$ is finite, $K_d\in L^2(\nu_d\otimes\nu_d)$; hence $T_d$ is Hilbert--Schmidt and therefore bounded and compact. Symmetry of $K_d$ and Fubini's theorem give self-adjointness. For $f\in L^2(I_d,\nu_d)$,
\[
  \langle T_df,f\rangle
  =\Var_{P_0}\!\left(
    \int_{I_d}f(j)\bigl(\mathbf1_{A_j}(X)-m(j)\bigr)
    \,d\nu_d(j)
  \right)\ge0.
\]
Thus $T_d$ is positive. Finally, the feature map
\[
  X\longmapsto
  \bigl(j\mapsto\mathbf1_{A_j}(X)-m(j)\bigr)
\]
has finite mean squared $L^2(\nu_d)$ norm, so its covariance operator is trace class.
\end{proof}

\begin{theorem}[Trace identity and eigenvalues]\label{thm:nonneg-eigen}
Let $(\lambda_k)_{k\ge1}$ be the nonzero eigenvalues of $T_d$, repeated according to multiplicity. Then
\[
  \lambda_k\ge0,
  \qquad
  \sum_{k=1}^\infty\lambda_k
  =\operatorname{tr}(T_d)
  =\int_{I_d}K_d(j,j)\,d\nu_d(j)
  =2^{-d}-3^{-d}.
\]
No tensor-product closed form for the eigenvalues is asserted for general $d$.
\end{theorem}

\begin{proof}
Compactness and self-adjointness give a discrete real spectrum, and positivity gives $\lambda_k\ge0$. Since $K_d(j,j)=m(j)(1-m(j))$,
\[
  \int_{I_d}K_d(j,j)\,d\nu_d(j)
  =2^{-d}\sum_c\left(\int F^c(p)\,dp-\int(F^c(p))^2\,dp\right).
\]
For every corner, the two integrals equal $2^{-d}$ and $3^{-d}$, respectively. Averaging over the $2^d$ corners yields $2^{-d}-3^{-d}$. The equality with the sum of the eigenvalues is the spectral trace identity for a positive trace-class operator.
\end{proof}

\section{Mean shifts under local density perturbations}\label{app:local-power}

This appendix records the deterministic mean-shift identities associated with local density perturbations. A full local-power theorem would additionally require triangular-array process convergence under the alternatives, which is not claimed here.

\begin{definition}[Contiguous local density perturbation]\label{def:local-density}
Let $h:[0,1]^d\to\R$ be bounded and measurable, with
\[
  \int_{[0,1]^d}h(y)\,dy=0.
\]
For all sufficiently large $n$, define the probability density
\[
  f_n(y) = 1 + \frac{1}{\sqrt{n}} h(y).
\]
Boundedness ensures $f_n\ge0$ for all sufficiently large $n$, while the zero-integral condition gives $\int f_n=1$.
\end{definition}

\begin{definition}[Corner mean shift function]\label{def:mean-shift}
For a corner $c\in\mathrm{Corner}(d)$ and $p\in[0,1]^d$, define
\[
  \mu_h^c(p)=\int_{[0,1]^d}\xi_{c,p}(y)h(y)\,dy,
\]
where $\xi_{c,p}(y) = \mathbf 1\{y \in \mathrm{cornerLe}(c,p)\} - F^c(p)$ is the centered corner indicator.
\end{definition}

\begin{theorem}[Expectation under local alternatives]\label{thm:local-alternative-expectation}
If $P_n$ has density $f_n$ with respect to Lebesgue measure, then
\[
  \E_{P_n}[\xi_{c,p}(Y)]
  =\frac{1}{\sqrt n}\mu_h^c(p).
\]
\end{theorem}

\begin{proof}
Expand $f_n=1+n^{-1/2}h$. The integral of $\xi_{c,p}$ under the uniform law is zero, leaving $n^{-1/2}\int\xi_{c,p}h=n^{-1/2}\mu_h^c(p)$.
\end{proof}

\begin{theorem}[Directional asymptotic mean shift vector]\label{thm:local-alternative-mean-shift-vector}
Let $\mathcal I$ be a finite index set, $(c_j)_{j\in\mathcal I}$ corners, $(p_j)_{j\in\mathcal I}$ evaluation points, $t \in \R^\mathcal I$ a direction vector, and $\mu_h = (\mu_h^{c_j}(p_j))_{j\in\mathcal I}$ the mean shift vector. Then
\[
  \E_{P_n}\!\left[\sum_{j\in\mathcal I}t_j\xi_{c_j,p_j}(Y)\right]
  =\frac{1}{\sqrt n}(t\cdot\mu_h).
\]
\end{theorem}

\begin{proof}
By linearity of finite sums and Theorem~\ref{thm:local-alternative-expectation}, the sum of expectations equals $n^{-1/2} \sum_{j\in\mathcal I} t_j \mu_h^{c_j}(p_j) = n^{-1/2} (t \cdot \mu_h)$.
\end{proof}

\begin{proposition}[Nonzero standardized directional shift]\label{thm:local-power-noncentrality}
If $t\cdot\mu_h\ne0$ and the null asymptotic variance $\sigma_t^2$ is positive, then
\[
  \delta^2 = \left( \frac{t \cdot \mu_h}{\sigma_t} \right)^2 > 0
\]
is strictly positive. This quantity identifies a shifted direction but, by itself, does not establish the local power of the quadratic statistic.
\end{proposition}

\begin{proof}
The quotient is nonzero and its square is therefore positive.
\end{proof}

\begin{remark}
To obtain a noncentral quadratic Gaussian limit for $n\bar\omega_d^2$ under $(P_n)$, one must establish joint or process-level convergence of the triangular-array empirical process to $G+\mu_h$ and then apply the continuous mapping theorem. Contiguity or a suitable Le Cam argument may provide such a route \cite{LeCam1986}; those steps are outside the results proved here.
\end{remark}

\end{appendices}

\bibliographystyle{plainnat}
\bibliography{sn-bibliography}

\end{document}